\documentclass{amsart}

\usepackage[a4paper,margin=2.5cm]{geometry}
\usepackage{amssymb}
\usepackage{tikz}
\usepackage{enumerate}
 \usepackage[pdfborder={0 0 0}]{hyperref}
\usepackage{hyperref}
\usepackage{color}
\theoremstyle{plain}
\newtheorem{thm}{Theorem}
\newtheorem{lem}[thm]{Lemma}

\newtheorem{cor}[thm]{Corollary}

\theoremstyle{definition}

\theoremstyle{remark}

\newtheorem{case}{Case}

\renewcommand{\geq}{\geqslant}
\renewcommand{\ge}{\geqslant}
\renewcommand{\leq}{\leqslant}
\renewcommand{\le}{\leqslant}

\newcommand{\M}{\mathrm{M}_{3}}
\newcommand{\MM}{\mathrm{M}}

\newcommand{\one}{$\mathrm{Anh}$}
\newcommand{\two}{$\mathrm{Bao}$}

\colorlet{grayred}{yellow!30!brown}
\colorlet{graygreen}{yellow!70!brown}
\colorlet{darkgreen}{black!45!green}

\usepackage{cancel}
\title{Two-player Tower of Hanoi}

\author[J.~Chappelon]{Jonathan Chappelon}
\address{Jonathan Chappelon, Institut de Math\'{e}matiques et de Modélisation de Montpellier, Universit\'{e} de Montpellier, Campus Triolet - CC051, place Eug\`{e}ne Bataillon, 34095 Montpellier Cedex 05, France}
\email{jonathan.chappelon@um2.fr}

\author[U.~Larsson]{Urban Larsson}
\address{Corresponding author: Urban Larsson, Department of Mathematics and Statistics, Dalhousie University, 6316 Coburg Road, PO Box 15000, Halifax, Nova Scotia, Canada B3H 4R2, supported by the Killam Trust}
\email{urban031@gmail.com}

\author[A.~Matsuura]{Akihiro Matsuura}
\address{Akihiro Matsuura, Division of Information System Design, School of Sience and Engineering, Tokyo Denki University, Ishizaka, Hatoyama-cho, Hiki, Saitama, 350-0394, Japan}
\email{matsu@rd.dendai.ac.jp}

\date{August 10, 2017}

\begin{document}
\begin{abstract}
The Tower of Hanoi game is a classical puzzle in recreational mathematics (Lucas 1883) which also has a strong record in pure mathematics. In a borderland between these two areas we find the characterization of the minimal number of moves, which is $2^n-1$, to transfer a tower of $n$ disks. But there are also other variations to the game, involving for example real number weights on the moves of the disks. This gives rise to a similar type of problem, but where the final score seeks to be optimized. We study extensions of the one-player setting to two players, invoking classical winning conditions in combinatorial game theory such as the player who moves last wins, or the highest score wins. Here we solve both these winning conditions on three heaps.
\end{abstract}
\maketitle
\section{Introduction/overview}
Tower of Hanoi (TH) is a classical puzzle invented by E. Lucas in 1883 (see e.g. \cite{Hi2, Sing} for mathematical, recreational, and historical surveys on this popular game); it is traditionally a `one player game' \cite{DH, DH09}. Here we let two players, \one\ (first player) and \two\ (second player), alternate turns and play a game on three or more pegs with various numbers of disks. We will begin by analyzing games under the following \emph{impartial} rules (the move options do not depend on who is to move) \cite{ANW, BCG, C}. 

Let $n\ge 1$ and $k\ge 3$ be positive integers. Two players alternate in transferring precisely one out of $n$ disks (of different sizes) on $k$ pegs, Peg $1,\ldots ,$ Peg $k$. The \emph{starting position} is as usual for the Tower of Hanoi; \emph{the tower} (i.e., all the disks in decreasing size) is placed on the \emph{starting peg} (Peg 1), and at each stage of the game, a larger disk cannot be placed on top of a smaller. In addition, the current player cannot move the disk that the previous player just moved (in Section~\ref{sec:disc} we discuss some variations to this rule). The game ends when the tower has been transferred to some predetermined \emph{final peg}, and some predetermined condition determines `who wins'. It is not allowed to transfer the tower to a non-final peg\footnote{Such positions could instead be declared drawn and we discuss some variation in Section~\ref{sec:disc}.}. We detail five ending conditions in Section~\ref{S:2}. A \emph{position} is any configuration of the $n$ disks on the $k$ pegs such that no larger disk is on top of a smaller disk, and in case of a non-starting position, a note on which disk was just moved.

Thus, a position $(S_1,\ldots ,S_k ; \alpha)$ is an ordered partitioning of the set $\{1, \ldots , n\}$ in $k$ sub-sets $$\bigcup_{i\in \{1,\ldots ,k\}} \!\!\!\! S_i = \{1, \ldots , n\}$$ such that $S_i\cap S_j = \emptyset$ for any $i\ne j$, and an integer $\alpha \in \{1, \ldots , k\}$. 

Let us exemplify our game with a position on seven disks and three pegs, and where the disk at the dashed peg has just been moved by the previous player and hence cannot be moved by the current player; below it we find its two legal options.\\

\begin{center}
\begin{tikzpicture}[scale=0.3]
\draw (0,0) -- (25,0);
\draw (4.5,0) -- (4.5,4);
\draw[dashed] (12.5,0) -- (12.5,4);
\draw (20.5,0) -- (20.5,4);
\draw[fill=graygreen] (1,0) rectangle (8,0.5);
\draw[fill=graygreen] (1.5,0.5) rectangle (7.5,1);
\draw[fill=graygreen] (2,1) rectangle (7,1.5);
\draw[fill=grayred] (10.5,0) rectangle (14.5,0.5);
\draw[fill=graygreen] (19,0) rectangle (22,0.5);
\draw[fill=graygreen] (19.5,0.5) rectangle (21.5,1);
\draw[fill=graygreen] (20,1) rectangle (21,1.5);
\end{tikzpicture}
 \end{center}

Move options:

\begin{center}
\begin{tikzpicture}[scale=0.3]
\draw (0,0) -- (25,0);
\draw (4.5,0) -- (4.5,4);
\draw[dashed] (12.5,0) -- (12.5,4);
\draw (20.5,0) -- (20.5,4);
\draw[fill=graygreen] (1,0) rectangle (8,0.5);
\draw[fill=graygreen] (1.5,0.5) rectangle (7.5,1);
\draw[fill=graygreen] (2,1) rectangle (7,1.5);
\draw[fill=graygreen] (10.5,0) rectangle (14.5,0.5);
\draw[fill=grayred] (12,.5) rectangle (13,1);
\draw[fill=graygreen] (19,0) rectangle (22,0.5);
\draw[fill=graygreen] (19.5,0.5) rectangle (21.5,1);
\end{tikzpicture}
\end{center}

\begin{center}
\begin{tikzpicture}[scale=0.3]
\draw (0,0) -- (25,0);
\draw[dashed] (4.5,0) -- (4.5,4);
\draw (12.5,0) -- (12.5,4);
\draw (20.5,0) -- (20.5,4);
\draw[fill=graygreen] (1,0) rectangle (8,0.5);
\draw[fill=graygreen] (1.5,0.5) rectangle (7.5,1);
\draw[fill=graygreen] (2,1) rectangle (7,1.5);
\draw[fill=grayred] (4,1.5) rectangle (5,2);
\draw[fill=graygreen] (10.5,0) rectangle (14.5,0.5);
\draw[fill=graygreen] (19,0) rectangle (22,0.5);
\draw[fill=graygreen] (19.5,0.5) rectangle (21.5,1);
\end{tikzpicture}
\end{center}
\vspace{.5 cm}
Henceforth it will be clear by the context which disk was moved last, so we omit the dashed emphasis. Note that, by this rule, there are exactly $3(2\cdot 3^{n-1}-1)=2(3^{n} - 3) + 3$ positions of two-player TH on $3$ pegs and $n$ disks; 2 positions for each configuration of the non-Tower configurations (depending on which disk was just moved) and the $3$ Tower configurations. A game is \emph{cyclic} (or loopy \cite{C2}) if there is at least one position which can be revisited during (not neccesarily optimal) play. Clearly TH is cyclic. 
\subsection{Normal play}
One winning condition is adapted from the one-player game: the player who plays the last disk (on top of the rest of the tower) wins. This corresponds to a classical convention for two-player games: a player who cannot move loses, which is called \emph{normal play} \cite{ANW, BCG, C, Si}. If no player can force a win in this setting, then the game is declared drawn.  

One of our first observations, in Section~\ref{S:2}, is that in spite of the cyclic nature of the game, if we play on just three pegs, \one\ wins the normal play convention. \two's moves will be forced throughout the game and the proof is an adaptation of the well known one-player result. In fact, \one's moves also have a restriction; she always has to move the smallest disk, but this turns out to be an advantage. In that section, we also note that the game is drawn on four or more pegs. 

\subsection{Scoring play}
Scoring two-player TH is the following game. The players move as in the normal play variation, but at each stage of game the current player gains a given real weight $w_{ij} = w_{ji}$, for a move from Peg~$i$ to Peg~$j$ (and the game starts with score 0 for each player). Thus, with three pegs there are three \emph{move edges}, $\{i, j\}\in\{\{1,2\}, \{1,3\},\{2,3\}\}$ and three real weights $w_{12},w_{13}$ and $w_{23}$; see Figure~\ref{F:1} for an illustration. For example, if \one's current score is $s$ and she moves along edge $\{i,j\}$, then her updated score is $s+w_{ij}$. The player who obtains the largest score when the game ends wins. We detail our terminology in Section~\ref{S:Scor}. 

If the game terminates but none of the players can claim a victory (because their terminal scores are equal) then the game is a tie. We also adapt the convention of drawn games from normal play, so a game is declared a draw if no player can force a win, by terminating the game\footnote{But, for future reference, we note that  there is also another choice for cyclic scoring games, if one of the players repeatedly moves to more negative scores than the other player, then, even though the game might not terminate, one could define it as a loss for the `more negative' player.}. 

Our main result is that \one\ wins nearly all scoring games on three pegs. The reason is partly the same as in normal play, but in this setting the optimal move sequence varies depending on the given weights, and the proof is non-trivial. The only case when \one\ cannot force a win is when all weights are equal and non-positive; otherwise she can attain an arbitrary high score by adhering to certain intermediate repetitive patterns. Questions of termination are often very hard (e.g. Turing machines), but in our three-peg setting it will be easy to distinguish drawn from winning. The first player, namely \one , controls all the moves under optimal play; if she cannot win, it will be easy to play drawn (see also Section~\ref{S:4}). The question of minimizing the number of moves in the two-player setting is studied in Section~\ref{S:6}.

Some backgorund: scoring play Tower of Hanoi is adapted from one-player Tower of Hanoi \cite{3, 1, 2}.  In \cite{3, 1}, the authors consider the recurrence relations generalized from the one by the Frame-Stewart algorithm\footnote{This algorithm produces an optimal path for four pegs \cite{Bousch} and is conjectured optimal for any $k > 4$.} for the $k$-peg Tower of Hanoi problem, by giving arbitrary positive integers as coefficients of the recurrences and obtained the exact formula for them. In \cite{2}, the authors consider another generalization for the three-peg Tower of Hanoi problem, where each undirected edge between pegs has a positive weight and the problem is to transfer all the disks from one peg to another with the minimum sum of weights, instead of the minimum number of moves, and obtain an optimal algorithm for that problem.

\section{Ending conditions for the two-player Tower of Hanoi}\label{S:2}

Consider the following variations of the two-player Tower of Hanoi; the game ends when the tower has been transferred to a \emph{final peg}, which is 
 
\begin{enumerate}[(EC1)]
\item\label{EC1}
a given peg, distinct from the starting peg;
\item\label{EC5}
the starting peg, but the largest disk has to be moved at least once;
\item\label{EC6}
the starting peg, but the smallest disk has to be moved at least once;
\item\label{EC3}
any peg, but the largest disk has to be moved at least once; and
\item\label{EC4}
any peg, but the smallest disk has to be moved at least once.
\end{enumerate}

Ending conditions (EC2) and (EC3) are not applicable when $n=1$. In Figure~\ref{F:1} to Figure~\ref{F:4}, we illustrate the idea of going from the one-player setting to the two-player setting using the standard graph representation. Each edge will now be directed, and the direction depends on the previous move.  In Figures~\ref{F:2} and ~\ref{F:4} we show which moves are possible in two-player Tower of Hanoi for $n=1$ and $2$ respectively (red dotted edges are illegal; green move edges are directed with direction depending on a given initial move). 
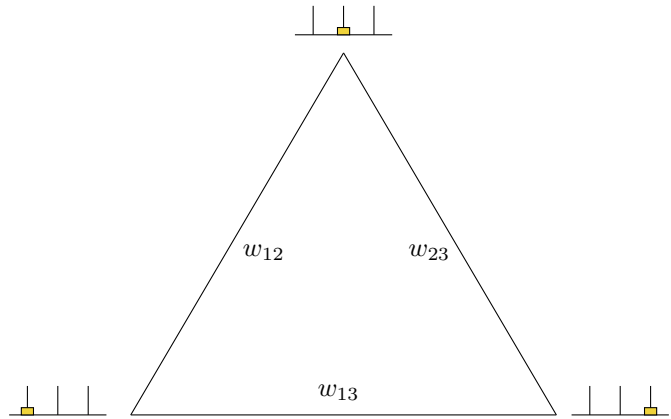
\begin{figure}
\begin{center}
\caption{A graph representation of Tower of Hanoi, with weighted move edges, for $n=1$.}\label{F:1}
\vspace{4mm}
\begin{tikzpicture}[xscale=0.4,yscale=0.48]
\draw (0,0) -- (7,10);
\draw (0,0) -- (14,0);
\draw (7,10) -- (14,0);
\begin{scope}[xshift=-4cm,yshift=-0cm,scale=0.4]
\draw (0,0) -- (8,0);
\draw (1.5,0) -- (1.5,2);
\draw (4,0) -- (4,2);
\draw (6.5,0) -- (6.5,2);
\draw[fill=graygreen] (1,0) rectangle (2,0.5);
\end{scope}
\begin{scope}[xshift=5.4cm,yshift=10.5cm, scale=0.4]
\draw (0,0) -- (8,0);
\draw (1.5,0) -- (1.5,2);
\draw (4,0) -- (4,2);
\draw (6.5,0) -- (6.5,2);
\draw[fill=graygreen] (3.5, 0) rectangle (4.5, 0.5);
\end{scope}
\begin{scope}[xshift=14.5cm,yshift=0cm, scale=0.4]
\draw (0,0) -- (8,0);
\draw (1.5,0) -- (1.5,2);
\draw (4,0) -- (4,2);
\draw (6.5,0) -- (6.5,2);
\draw[fill=graygreen] (6, 0) rectangle (7, 0.5);
\put(42,60){$w_{12}$}
\put(70,7){$w_{13}$}
\put(104,60){$w_{23}$}
\end{scope}
\end{tikzpicture}
\end{center}
\end{figure}
\begin{figure}
\begin{center}
\caption{ A play of two-player TH for $n=1$, (EC1).}\label{F:2}
\vspace{3 mm}
\begin{tikzpicture}[xscale=0.4,yscale=0.48]
\draw [red, thick, dotted] (0,0) -- (7,10);
\draw [darkgreen, thick, ->] (0,0) -- (14,0);
\draw [red, thick, dotted] (7,10) -- (14,0);
\begin{scope}[xshift=-4cm,yshift=-0cm,scale=0.4]
\draw (0,0) -- (8,0);
\draw (1.5,0) -- (1.5,2);
\draw (4,0) -- (4,2);
\draw (6.5,0) -- (6.5,2);
\draw[fill=graygreen] (1,0) rectangle (2,0.5);
\end{scope}
\begin{scope}[xshift=5.4cm,yshift=10.5cm, scale=0.4]
\draw (0,0) -- (8,0);
\draw (1.5,0) -- (1.5,2);
\draw (4,0) -- (4,2);
\draw (6.5,0) -- (6.5,2);
\draw[fill=graygreen] (3.5, 0) rectangle (4.5, 0.5);
\end{scope}
\begin{scope}[xshift=14.5cm,yshift=0cm, scale=0.4]
\draw (0,0) -- (8,0);
\draw (1.5,0) -- (1.5,2);
\draw (4,0) -- (4,2);
\draw (6.5,0) -- (6.5,2);
\draw[fill=graygreen] (6, 0) rectangle (7, 0.5);
\end{scope}
\end{tikzpicture}
\end{center}
\end{figure}
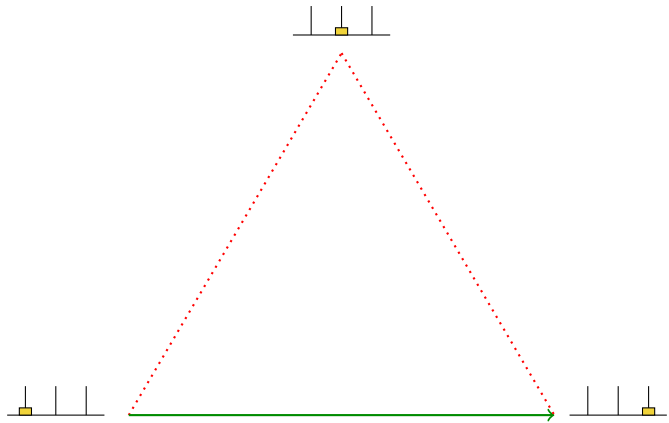
\vspace{3 mm}
\begin{figure}
\begin{center}
\caption{A graph representation of one-player TH for $n=2$.}\label{F:3}
\vspace{3 mm}
\begin{tikzpicture}[xscale=0.4,yscale=0.47]
\draw (0,0) -- (7,10);
\draw (0,0) -- (14,0);
\draw (7,10) -- (14,0);
\draw (4.6, 6.6) -- (9.4, 6.6);
\draw (4.6, 0) -- (2.3, 3.3);
\draw (9.4, 0) -- (11.7, 3.3);
\begin{scope}[xshift=-1.7cm,yshift=3.3cm,scale=0.4]
\draw (0,0) -- (8,0);
\draw (1.5,0) -- (1.5,2);
\draw (4,0) -- (4,2);
\draw (6.5,0) -- (6.5,2);
\draw[fill=graygreen] (0.5,0) rectangle (2.5,0.5);
\draw[fill=graygreen] (6,0) rectangle (7,0.5);
\end{scope}
\begin{scope}[xshift= 0.6cm,yshift=6.6cm,scale=0.4]
\draw (0,0) -- (8,0);
\draw (1.5,0) -- (1.5,2);
\draw (4,0) -- (4,2);
\draw (6.5,0) -- (6.5,2);
\draw[fill=graygreen] (3,0) rectangle (5,0.5);
\draw[fill=graygreen] (6,0) rectangle (7,0.5);
\end{scope}
\begin{scope}[xshift=12.6cm,yshift=3.3cm,scale=0.4]
\draw (0,0) -- (8,0);
\draw (1.5,0) -- (1.5,2);
\draw (4,0) -- (4,2);
\draw (6.5,0) -- (6.5,2);
\draw[fill=graygreen] (1,0) rectangle (2,0.5);
\draw[fill=graygreen] (5.5,0) rectangle (7.5,0.5);
\end{scope}
\begin{scope}[xshift=10.2cm,yshift=6.6cm,scale=0.4]
\draw (0,0) -- (8,0);
\draw (1.5,0) -- (1.5,2);
\draw (4,0) -- (4,2);
\draw (6.5,0) -- (6.5,2);
\draw[fill=graygreen] (3,0) rectangle (5,0.5);
\draw[fill=graygreen] (1,0) rectangle (2,0.5);
\end{scope}
\begin{scope}[xshift=-4cm,yshift=-0cm,scale=0.4]
\draw (0,0) -- (8,0);
\draw (1.5,0) -- (1.5,2);
\draw (4,0) -- (4,2);
\draw (6.5,0) -- (6.5,2);
\draw[fill=graygreen] (0.5,0) rectangle (2.5,0.5);
\draw[fill=graygreen] (1,0.5) rectangle (2,1);
\end{scope}
\begin{scope}[xshift=3cm,yshift=-1.5cm,scale=0.4]
\draw (0,0) -- (8,0);
\draw (1.5,0) -- (1.5,2);
\draw (4,0) -- (4,2);
\draw (6.5,0) -- (6.5,2);
\draw[fill=graygreen] (0.5,0) rectangle (2.5,0.5);
\draw[fill=graygreen] (3.5,0) rectangle (4.5,0.5);
\end{scope}
\begin{scope}[xshift=7.8cm,yshift=-1.5cm,scale=0.4]
\draw (0,0) -- (8,0);
\draw (1.5,0) -- (1.5,2);
\draw (4,0) -- (4,2);
\draw (6.5,0) -- (6.5,2);
\draw[fill=graygreen] (3.5,0) rectangle (4.5,0.5);
\draw[fill=graygreen] (5.5,0) rectangle (7.5,0.5);
\end{scope}
\begin{scope}[xshift=5.4cm,yshift=10.5cm, scale=0.4]
\draw (0,0) -- (8,0);
\draw (1.5,0) -- (1.5,2);
\draw (4,0) -- (4,2);
\draw (6.5,0) -- (6.5,2);
\draw[fill=graygreen] (3.5,0.5) rectangle (4.5,1);
\draw[fill=graygreen] (3,0) rectangle (5,0.5);
\end{scope}
\begin{scope}[xshift=14.5cm,yshift=0cm, scale=0.4]
\draw (0,0) -- (8,0);
\draw (1.5,0) -- (1.5,2);
\draw (4,0) -- (4,2);
\draw (6.5,0) -- (6.5,2);
\draw[fill=graygreen] (5.5,0) rectangle (7.5,0.5);
\draw[fill=graygreen] (6,0.5) rectangle (7,1);
\end{scope}
\end{tikzpicture}
\end{center}
\end{figure}
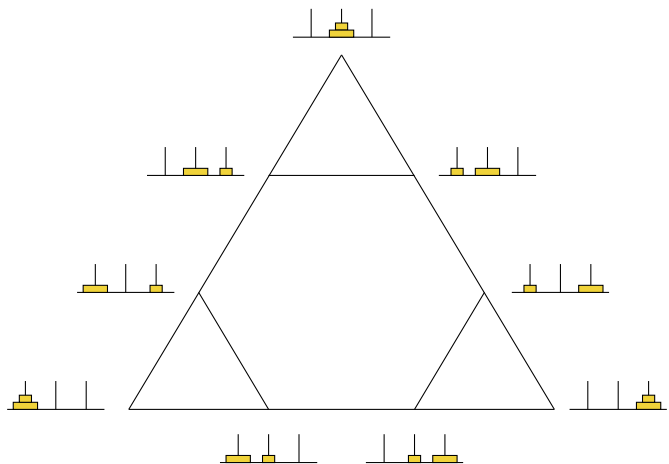
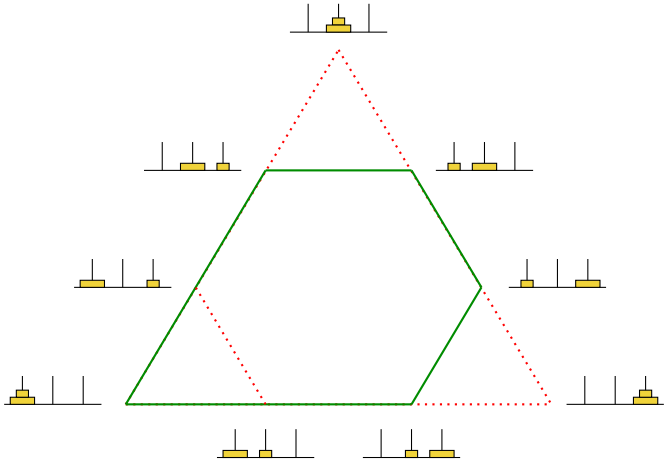
\begin{figure}
\begin{center}
\caption{A play sequence of two-player TH for $n=2$, (EC2); depicting the first row in Table~\ref{table:1}.}\label{F:4} 
\vspace{4 mm}
\begin{tikzpicture}[xscale=0.4,yscale=0.47]
\draw [red, thick, dotted] (0,0) -- (7,10);
\draw [red, thick, dotted] (0,0) -- (14,0);
\draw [red, thick, dotted] (7,10) -- (14,0);
\draw [darkgreen, thick] (0,0) -- (4.6, 6.6);
\draw [darkgreen, thick] (0,0) -- (9.4, 0);
\draw [darkgreen, thick] (11.7, 3.3) -- (9.4, 6.6);
\draw [darkgreen, thick] (4.6, 6.6) -- (9.4, 6.6);
\draw [red, thick, dotted] (4.6, 0) -- (2.3, 3.3);
\draw [darkgreen, thick] (9.4, 0) -- (11.7, 3.3);
\begin{scope}[xshift=-1.7cm,yshift=3.3cm,scale=0.4]
\draw (0,0) -- (8,0);
\draw (1.5,0) -- (1.5,2);
\draw (4,0) -- (4,2);
\draw (6.5,0) -- (6.5,2);
\draw[fill=graygreen] (0.5,0) rectangle (2.5,0.5);
\draw[fill=graygreen] (6,0) rectangle (7,0.5);
\end{scope}
\begin{scope}[xshift= 0.6cm,yshift=6.6cm,scale=0.4]
\draw (0,0) -- (8,0);
\draw (1.5,0) -- (1.5,2);
\draw (4,0) -- (4,2);
\draw (6.5,0) -- (6.5,2);
\draw[fill=graygreen] (3,0) rectangle (5,0.5);
\draw[fill=graygreen] (6,0) rectangle (7,0.5);
\end{scope}
\begin{scope}[xshift=12.6cm,yshift=3.3cm,scale=0.4]
\draw (0,0) -- (8,0);
\draw (1.5,0) -- (1.5,2);
\draw (4,0) -- (4,2);
\draw (6.5,0) -- (6.5,2);
\draw[fill=graygreen] (1,0) rectangle (2,0.5);
\draw[fill=graygreen] (5.5,0) rectangle (7.5,0.5);
\end{scope}
\begin{scope}[xshift=10.2cm,yshift=6.6cm,scale=0.4]
\draw (0,0) -- (8,0);
\draw (1.5,0) -- (1.5,2);
\draw (4,0) -- (4,2);
\draw (6.5,0) -- (6.5,2);
\draw[fill=graygreen] (3,0) rectangle (5,0.5);
\draw[fill=graygreen] (1,0) rectangle (2,0.5);
\end{scope}
\begin{scope}[xshift=-4cm,yshift=-0cm,scale=0.4]
\draw (0,0) -- (8,0);
\draw (1.5,0) -- (1.5,2);
\draw (4,0) -- (4,2);
\draw (6.5,0) -- (6.5,2);
\draw[fill=graygreen] (0.5,0) rectangle (2.5,0.5);
\draw[fill=graygreen] (1,0.5) rectangle (2,1);
\end{scope}
\begin{scope}[xshift=3cm,yshift=-1.5cm,scale=0.4]
\draw (0,0) -- (8,0);
\draw (1.5,0) -- (1.5,2);
\draw (4,0) -- (4,2);
\draw (6.5,0) -- (6.5,2);
\draw[fill=graygreen] (0.5,0) rectangle (2.5,0.5);
\draw[fill=graygreen] (3.5,0) rectangle (4.5,0.5);
\end{scope}
\begin{scope}[xshift=7.8cm,yshift=-1.5cm,scale=0.4]
\draw (0,0) -- (8,0);
\draw (1.5,0) -- (1.5,2);
\draw (4,0) -- (4,2);
\draw (6.5,0) -- (6.5,2);
\draw[fill=graygreen] (3.5,0) rectangle (4.5,0.5);
\draw[fill=graygreen] (5.5,0) rectangle (7.5,0.5);
\end{scope}
\begin{scope}[xshift=5.4cm,yshift=10.5cm, scale=0.4]
\draw (0,0) -- (8,0);
\draw (1.5,0) -- (1.5,2);
\draw (4,0) -- (4,2);
\draw (6.5,0) -- (6.5,2);
\draw[fill=graygreen] (3.5,0.5) rectangle (4.5,1);
\draw[fill=graygreen] (3,0) rectangle (5,0.5);
\end{scope}
\begin{scope}[xshift=14.5cm,yshift=0cm, scale=0.4]
\draw (0,0) -- (8,0);
\draw (1.5,0) -- (1.5,2);
\draw (4,0) -- (4,2);
\draw (6.5,0) -- (6.5,2);
\draw[fill=graygreen] (5.5,0) rectangle (7.5,0.5);
\draw[fill=graygreen] (6,0.5) rectangle (7,1);
\end{scope}
\end{tikzpicture}
\end{center}
\end{figure}
\vspace{10 mm}
\section{General play}\label{S:Gen}

In this section we regard `an odd or even number of moves' in the one-player game in the same sense as for the two-player games; no disk will be moved twice unless some other disk has been moved in between. Each odd numbered move thus moves the smallest disk, and in fact we could equivalently have chosen to let \one\ lead the game. Recall that the starting position throughout the paper is all disks on Peg~1. The following two results are easy well-known consequences of the graph representation of Tower of Hanoi; we complement with proofs by induction.

\begin{thm}\label{thm:1}
For three pegs, $n\ge 2$ disks can be transferred from the starting position to any position using an odd number of moves. That is, in the two-player version, the first player can force play to any position in an odd number of moves. 
\end{thm}

\begin{proof}
We use induction on $n$ and the fact that the smallest disk is always moved on the odd numbered moves. For $n=2$, Table~\ref{table:1} gives examples of sequences for transferring the two disks from Peg~1 to each of the nine possible positions, using in each case an odd number of moves. Suppose now that the result is true for $n\ge 2$ disks. We want to transfer $n+1$ disks from the starting position, say from Peg~1, to a position $P$. We distinguish different cases depending on the position of the largest disk LD in $P$. If LD is on Peg~1 in $P$, we have just to transfer the $n$ other disks from Peg~1 to the position $P$ and we do it by using an odd number of moves by induction hypothesis. Otherwise, if LD is on, say Peg~2 in $P$, we transfer the $n$ smallest disks from Peg~1 to Peg~3 using $n_1$ moves, the LD from Peg~1 to Peg~2, and finally the $n$ smallest disks from Peg~3 to the position $P$ using $n_2$ moves. Thus, we have transferred the $n+1$ disks from Peg~1 to the position $P$ using $n_1+n_2+1$ moves and, since $n_1$ and $n_2$ are odd by induction hypothesis, this number of moves is odd. If LD is on peg 3, the argument is the same. This completes the proof.
\end{proof}

\begin{table}
\begin{center}
\caption{Any Tower of Hanoi position  on two disks and three pegs can be reached in an odd number of moves.}\label{table:1}
\begin{tabular}{|c|c|c|}
\hline
Position & Sequence of moves & Number of moves \\
\hline
\begin{tikzpicture}[scale=0.2]
\draw (0,0) -- (8,0);
\draw (1.5,0) -- (1.5,2);
\draw (4,0) -- (4,2);
\draw (6.5,0) -- (6.5,2);
\draw[fill=graygreen] (0.5,0) rectangle (2.5,0.5);
\draw[fill=graygreen] (1,0.5) rectangle (2,1);
\end{tikzpicture}
& $13 - 12 - 13 - 23 - 12 - 13 - 12$ &  $7$ \\
\hline
\begin{tikzpicture}[scale=0.2]
\draw (0,0) -- (8,0);
\draw (1.5,0) -- (1.5,2);
\draw (4,0) -- (4,2);
\draw (6.5,0) -- (6.5,2);
\draw[fill=graygreen] (0.5,0) rectangle (2.5,0.5);
\draw[fill=graygreen] (3.5,0) rectangle (4.5,0.5);
\end{tikzpicture}
& $12$ &  $1$ \\
\hline
\begin{tikzpicture}[scale=0.2]
\draw (0,0) -- (8,0);
\draw (1.5,0) -- (1.5,2);
\draw (4,0) -- (4,2);
\draw (6.5,0) -- (6.5,2);
\draw[fill=graygreen] (0.5,0) rectangle (2.5,0.5);
\draw[fill=graygreen] (6,0) rectangle (7,0.5);
\end{tikzpicture}
& $13$ &  $1$ \\
\hline
\begin{tikzpicture}[scale=0.2]
\draw (0,0) -- (8,0);
\draw (1.5,0) -- (1.5,2);
\draw (4,0) -- (4,2);
\draw (6.5,0) -- (6.5,2);
\draw[fill=graygreen] (1,0) rectangle (2,0.5);
\draw[fill=graygreen] (3,0) rectangle (5,0.5);
\end{tikzpicture}
& $13 - 12 - 13$ &  $3$ \\
\hline
\begin{tikzpicture}[scale=0.2]
\draw (0,0) -- (8,0);
\draw (1.5,0) -- (1.5,2);
\draw (4,0) -- (4,2);
\draw (6.5,0) -- (6.5,2);
\draw[fill=graygreen] (3.5,0.5) rectangle (4.5,1);
\draw[fill=graygreen] (3,0) rectangle (5,0.5);
\end{tikzpicture}
& $13 - 12 - 23$ &  $3$ \\
\hline
\begin{tikzpicture}[scale=0.2]
\draw (0,0) -- (8,0);
\draw (1.5,0) -- (1.5,2);
\draw (4,0) -- (4,2);
\draw (6.5,0) -- (6.5,2);
\draw[fill=graygreen] (3,0) rectangle (5,0.5);
\draw[fill=graygreen] (6,0) rectangle (7,0.5);
\end{tikzpicture}
& $12 - 13 - 12 - 23 - 13$ &  $5$ \\
\hline
\begin{tikzpicture}[scale=0.2]
\draw (0,0) -- (8,0);
\draw (1.5,0) -- (1.5,2);
\draw (4,0) -- (4,2);
\draw (6.5,0) -- (6.5,2);
\draw[fill=graygreen] (1,0) rectangle (2,0.5);
\draw[fill=graygreen] (5.5,0) rectangle (7.5,0.5);
\end{tikzpicture}
& $12 - 13 - 12$ &  $3$ \\
\hline
\begin{tikzpicture}[scale=0.2]
\draw (0,0) -- (8,0);
\draw (1.5,0) -- (1.5,2);
\draw (4,0) -- (4,2);
\draw (6.5,0) -- (6.5,2);
\draw[fill=graygreen] (3.5,0) rectangle (4.5,0.5);
\draw[fill=graygreen] (5.5,0) rectangle (7.5,0.5);
\end{tikzpicture}
& $13 - 12 - 13 - 23 - 12$ &  $5$ \\
\hline
\begin{tikzpicture}[scale=0.2]
\draw (0,0) -- (8,0);
\draw (1.5,0) -- (1.5,2);
\draw (4,0) -- (4,2);
\draw (6.5,0) -- (6.5,2);
\draw[fill=graygreen] (5.5,0) rectangle (7.5,0.5);
\draw[fill=graygreen] (6,0.5) rectangle (7,1);
\end{tikzpicture}
& $12 - 13 - 23$ &  $3$ \\
\hline
\end{tabular}
\end{center}
\end{table}

\begin{cor}\label{cor:1}
For three pegs, $n\ge 2$ disks can be transferred from the starting position to any intermediate position using an even number of moves. 
\end{cor}

\begin{proof}
Let $P$ be an intermediate position and suppose that the smallest disk is on Peg~p. Since $P$ is an intermediate position, we know that one of the two other pegs, Peg~q or Peg~r, contains at least one disk. Let QD be the smallest disk of the disks on Peg~q and Peg~r and suppose WLOG that QD is on Peg~q. Let $P'$ be the position where all the disks are positioned like in $P$, except QD which is on Peg~r, instead of Peg~q in $P$. Since all the disks can be transferred from Peg~1 to the position $P'$ using an odd number moves by Theorem~\ref{thm:1}, we reach the position $P$ by only adding one move of transferring QD from Peg~r to Peg~q. Indeed, the smallest disk was moved in the previous move to achieve the position $P'$. Thus, we have transferred all the disks from Peg~p to the position $P$ using an even number of moves.
\end{proof}

\section{Normal play: two-player Tower of Hanoi}

In the normal play variation of the two-player Tower of Hanoi, \one\  can avoid drawn simply by adhering to the well known \emph{minimal algorithm} for the one-player Tower of Hanoi (\two's moves will be forced all through the game), using precisely $2^n-1$ moves. However, she can also choose freely among all odd-length move paths. 

\begin{cor}\label{cor:2}
For three pegs and $n\ge 1$ disks, the two-player Tower of Hanoi game terminates and the first player wins. This is true for any ending condition and also from any position, provided that the previous player did not move the smallest disk.
\end{cor}

\begin{proof}
From Theorem~\ref{thm:1}, the number of moves for the one-player Tower of Hanoi can be odd. We use any such sequence in the two-player game. In the one-player game, every odd move transfers the smallest of the $n$ disks, which here means that, in each move,  \one\ will move the smallest of the $n$ disks on top of an empty peg or a larger disk. This forces \two\ to move a larger disk at each stage of the game. Precisely one such move is possible since \one's move of the smallest disk occupies one out of precisely three pegs. Since the number of moves used here is odd, \one\ wins. 
\end{proof}

Games on four pegs are mostly loopy.

\begin{thm}\label{thm:2} 
The two-player Tower of Hanoi on four or more pegs is a draw game if the number of disks is three or more, given any ending condition. 
\end{thm}

\begin{proof} 
Suppose that there are three or more disks. Then \two's moves are never forced; he never has to place the second smallest disk on top of the third smallest (analogously for \one). 
\end{proof}

For completeness, let us also give the rest of the $k$-peg observations with $k\ge 4$. For (EC1-3), if there are two disks, \two\ never has to move the largest disk to a final peg and hence the game is drawn. For (EC1), if there is only one disk, then \one\ wins in the first move. For (EC2,3), if there is only one disk, the special rule is invoked and \two\ wins in his first move. For (EC4,5), if there are two disks, \two\ has to move the largest disk to a final peg and hence loses. If there is only one disk, then \one\ wins in the first move.

\section{Scoring play: two-player Tower of Hanoi with weights}\label{S:Scor}

As stated in the introduction, for the scoring variation of the normal play setting, we provide real weights to the \emph{move edges}, in the three-peg case, $w_{12}, w_{13}$ and $w_{23}$ respectively. As usual, the two players alternate in moving, and a player gets $w_{ij}$ points for a move along edge $\{i,j\}$. The player who has most points when the game ends wins. We begin by giving the solution of the game with less than three disks.

We will use $A_{ij}(n)$ and $B_{ij}(n)$, for the total points for \one\ and \two\ respectively, of the two-player Tower of Hanoi game, for transferring $n$ disks from Peg~$i$ to Peg~$j$ by a given algorithm, for example the minimal algorithm, and we let the \emph{total score} be $\Delta_{ij}(n) = A_{ij}(n)-B_{ij}(n)$, or just $\Delta(n)$. Hence, if $\Delta(n)>0$ then \one\ wins, if $\Delta(n)=0$, then the game is tie, and otherwise \two\ wins.

\begin{thm}\label{thm:3}
Consider the two-player Tower of Hanoi game on two disks, three pegs and three weights of real numbers $w_{12}$, $w_{13}$, and $w_{23}$. In case of (EC4) or (EC5), the first player wins if either of the following inequalities holds:
\begin{align}
w_{12}+w_{23}-w_{13} > 0\label{tag1}\\ 
3w_{13} - w_{12} - w_{23} > 0\label{tag2}\\
w_{13}+w_{23}-w_{12} > 0\label{tag3}\\ 
3w_{12} - w_{13} - w_{23} > 0\label{tag4}\\
w_{12}+w_{13}-w_{23} > 0\label{tag5}
\end{align}
In case of (EC2) or (EC3), she wins if (\ref{tag5}) holds. In case of (EC1), she wins if (\ref{tag1}) or (\ref{tag2}) holds. Otherwise the game is a draw. 

If the game is played on only one disk, then (EC2,3) are not applicable. For (EC1) the first player wins if $w_{13}>0$; loses if $w_{13}<0$; and the game is a tie otherwise. The second player wins (EC4,5) if $w_{12}<0$ and $w_{13}<0$. The game is a tie if at least one of these weights is 0 and the others are non-positive. Otherwise the first player wins. 
\end{thm}

\begin{proof}
Notice that in this game, \one\ will only move the small disk and \two\ will only move the large one. There are only six possibilities. If \one\ is not able to force a win, then she will resort to drawn, using either of the two possibilities:
$$
12-(13-12-23)^{\infty}
$$
or
$$
13-(12-13-23)^{\infty},
$$
where $ij$ denotes the current player's  move between Peg~$i$ to Peg~$j$, and where $(\cdot )^\infty$ denotes an infinite repetition of a given move sequence. But, in case it is to her advantage, she can interrupt either of these two sequences of moves, and use either of the following six sequences: 
$$12-(13-12-23)^{2k}-13-23$$ 
$$13-(12-13-23)^{2k+1}-13$$
$$12-(13-12-23)^{2k+1}-12$$
$$13-(12-13-23)^{2k}-12-23$$
$$12-(13-12-23)^{2k+1}-13-12-13$$
$$13-(12-13-23)^{2k+1}-12-13-12$$
for $k$ a non-negative integer. In the first two cases she will end the game on Peg~3 for (EC1,4,5), corresponding to the cases (\ref{tag1}) and (\ref{tag2}) respectively. For the two middle move sequences, she will end on Peg~2 for (EC4,5), corresponding to the cases (\ref{tag3}) and (\ref{tag4}) respectively. To terminate the game on Peg~1, only valid for (EC2-5), she uses one of the two last move sequences. They are symmetric and both result in the inequality (\ref{tag5}). In either case, the largest disk has been moved.

In each of the above cases we evaluate the value of $A(2)-B(2)$ and then the triangular inequalities appear, and it is clear that the difference of total points is independent of the choice of $k$ in each case. 
\end{proof}

\begin{figure}[h!]
\begin{center}
\caption{Consider (EC1) for $n=2$. Here $w_{23} = -3$ and the other weights represent the $x$- and $y$-axes.  The game is drawn in the white area. Compare this picture with the result for  $n\geqslant 3$ in Theorem~\ref{thm:4}, where the class of drawn games would have been represented by a single white dot at $(-3,-3)$.}\label{fig:4}
\vspace{2mm}
\includegraphics[width=0.4\textwidth]{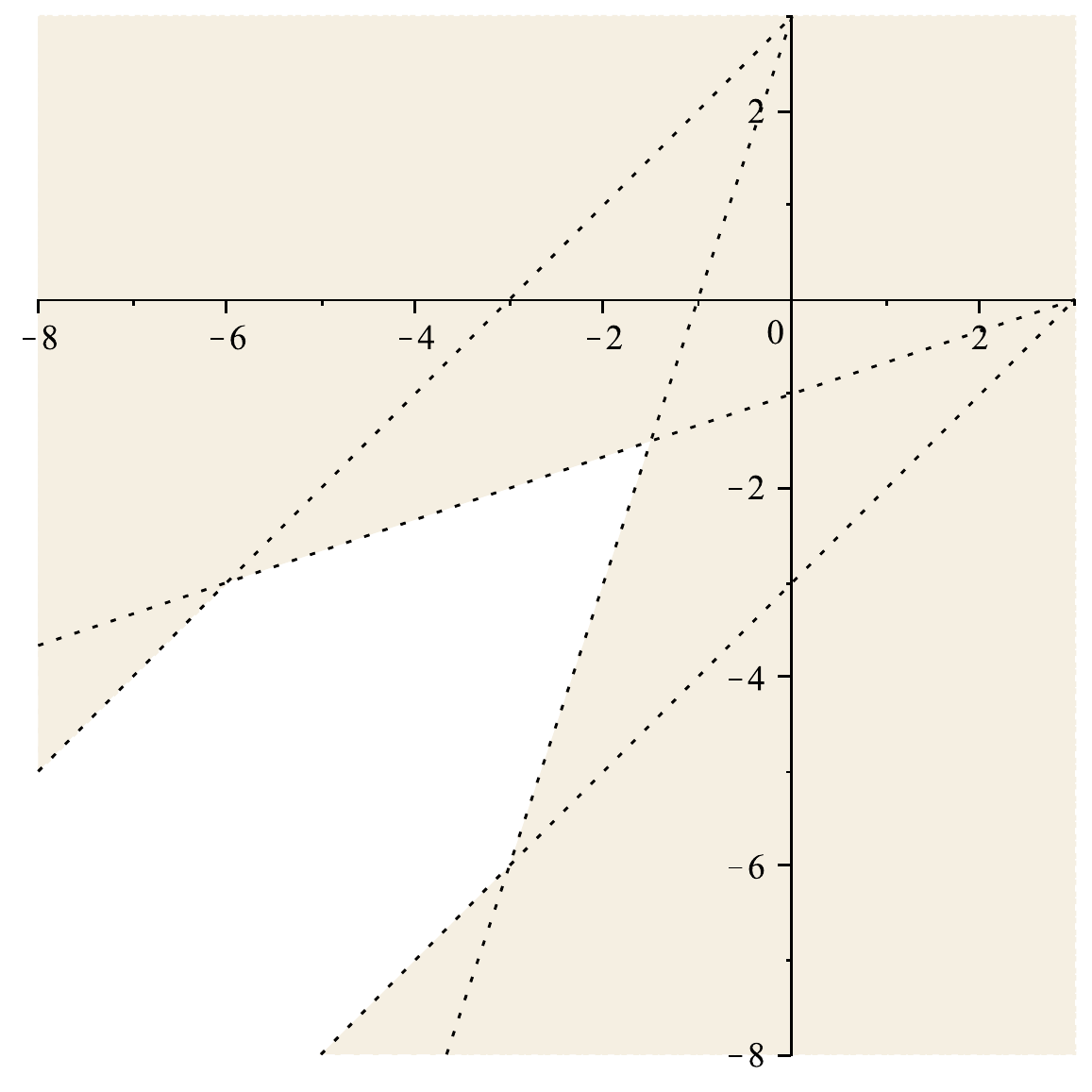}
\end{center}
\end{figure}

It turns out that the case $n\ge 3$ has fewer drawn games, allowing a simpler description; it relies on the general ideas in Section~\ref{S:2}. By relabeling the pegs, it suffices to analyze the case of transferring the disks from Peg~1 to Peg~3  and the case from Peg~1 to Peg~1 .

\begin{thm}\label{thm:4}
Given $n\ge 3$ disks, three pegs and three weights of real numbers $w_{13}$, $w_{12}$, and $w_{23}$. Then, for the two-player Tower of Hanoi game, the first player wins every game, except in the case $w_{12}=w_{13}=w_{23}\leq 0$ for which the game is a draw.
\end{thm}
 
\begin{proof}
If $w_{12}=w_{13}=w_{23}=\alpha$, it is easy to see that, with any strategy, if the game is not drawn (because of infinite play), then the total score is $\Delta(n)=\alpha$ for any $n\geq 1$ because the total number of moves is always odd if the game ends. The result follows in this case (because if $\alpha < 0$, then \one\ need not end the game so it is drawn).

Let $\{i,j,k\}=\{1,2,3\}$ and suppose that $w_{ij}$ is the smallest number among $w_{ij}$, $w_{ik}$ and $w_{jk}$. Then, we have $w_{ik}+w_{jk}>2w_{ij}$. Let $P$ be an intermediate position where the two smallest disks are on Peg~$k$ and let the smallest disk among the disks on Peg~$i$ and Peg~$j$ be on Peg~$i$. We know, from Corollary~\ref{cor:1}, that the $n\ge3$ disks can be transferred from the starting position on Peg~1 to the intermediate position $P$ using an even sequence of moves $s_1$ and, from Theorem~\ref{thm:1}, that all the disks can be transferred from the final position on Peg~f to the intermediate position $P$ using an odd sequence of moves $s_2$. Then, it is clear that the reverse sequence of $s_2$ is an odd sequence ${s_2}^{-1}$ of moves transferring all the disks from the intermediate position $P$ to the final position. Thus, the sequence of moves $s_1{s_2}^{-1}$ is a finite strategy for transferring $n\ge3$ disks from Peg~1 to Peg~f. Let $p$ be the total score $\Delta(n)$ associated with this strategy. If $p>0$, \one\ wins following this strategy and the result is then proved. Suppose now that $p\le0$. We consider a new strategy where \one\ forces \two\ to play as follows:
\begin{enumerate}[1)]
\item
\one\ starts the game and forces \two\ to play the even number of moves of $s_1$, where the $n$ disks are transferred from the starting position on Peg~1 to the intermediate position $P$. After that, since $s_1$ is even, it is \one's turn.
\item
\one\ continues by forcing \two\ to play an even number of moves, whose sequence is denoted by $s_3$ and for which the total score $\Delta(n)$ is incremented by a positive value and where the $n$ disks are always at the position $P$ after $s_3$. Note that after this, it is always \one's turn since $s_3$ is even. This step is repeated until the total score $\Delta(n)$ is sufficiently large.
\item
\one\ finishes the game by forcing \two\ to play the odd number of moves of ${s_2}^{-1}$, where the $n$ disks are transferred from the intermediate position $P$ to the final position on Peg~f.
\end{enumerate}
It remains to detail the sequence $s_3$ and to verify that it effectively increments the total score $\Delta(n)$ by a positive value. In fact, the sequence $s_3$ consists of the 7 moves that transfer the two smallest disks from Peg~$k$ to Peg~$k$, as already seen in the proof of Theorem~\ref{thm:1}, and the move of the smallest disk of Peg~$i$ and Peg~$j$ from Peg~$i$ to Peg~$j$, and we repeat these $8$ moves twice so that we return to the intermediate position $P$. There are two possible sequences for $s_3$, that are
$$
\left(ik - jk - ik - ij - jk - ik - jk - ij\right)^2\ \text{or}\ \left( jk - ik - jk - ij - ik - jk - ik - ij \right)^2.
$$
In both cases, the total score $\Delta(n)$ is incremented by $2\left(w_{ik}+w_{jk}-2w_{ij}\right)$, which is strictly positive by hypothesis. So we have proved that \one\ wins the game by using the following strategy $s_1{s_3}^\lambda{s_2}^{-1}$, where
$$
\lambda = \left\lfloor\frac{-p}{2\left(w_{ik}+w_{jk}-2w_{ij}\right)}\right\rfloor+1.
$$
This completes the proof.
\end{proof}

We can adapt this proof to scoring play from an arbitrary position of disks for the three-peg case. But we have to remember that in the two-player setting a position carries also a memory of the last move.

\begin{cor}
Consider an arbitrary Tower of Hanoi position on $n\geq 3$ disks, three pegs and three weights of real numbers $w_{13}$, $w_{12}$, and $w_{23}$. Then the first player wins, unless $w_{12}=w_{13}=w_{23}\leq 0$, and provided the previous player did not move the smallest disk. 
\end{cor}

\begin{proof}
Apply Theorem~\ref{thm:1} and the proof of Theorem~\ref{thm:4} (omitting $s_1$).
\end{proof}

\section{The minimal number of moves for winning}\label{S:6}

\subsection{The minimal number of moves for winning normal play}

Traditionally, in the one-player setting, the interest has often been focused on the minimal number of moves for transferring the tower. In this section we analyze our variations of the two-player game in this sense. It is not a big surprise that the minimal number of moves to win normal play is the same as the number of moves in the one-player minimal algorithm, but let us sum up the state of the art before we move on to the more challenging analysis of minimum number of moves for winning scoring play. 

\begin{thm}\label{thm:5}
The minimum number of moves for transferring $n\ge1$ disks from one peg to another peg is $2^{n}-1$. The minimum number of moves for transferring $n\ge2$ disks from one peg to the same peg is $2^{n+1}-1$, if the largest disk has to be moved; and it is seven if only the smallest disk has to be moved.
\end{thm}

\begin{proof}
If we want to transfer $n\ge 1$ disks from Peg~1 to Peg~3, it is well known that the minimum number of moves needed is exactly $2^{n}-1$. Recall here in few words how to obtain this result. We prove it by induction on $n$. For $n=1$, the result is clear. Now, suppose that the result is true for transferring $n-1$ disks from one peg to another peg. If we want to transfer the $n$th disk from Peg~1 to Peg~3, the $n-1$ smallest disks must be on Peg~2. So, we transfer the $n-1$ smallest disks from Peg~1 to Peg~2 using $2^{n-1}-1$ moves by induction hypothesis, the largest disk from Peg~1 to Peg~3 and finally the $n-1$ smallest disks from Peg~2 to Peg~3 using also $2^{n-1}-1$ moves by induction hypothesis. This is the reason why the minimal number of moves needed for transferring $n$ disks from Peg~1 to Peg~3 is $2^{n}-1$.

Now, we want to transfer $n\ge 2$ disks from Peg~1 to Peg~1 with the condition that all disks have been moved; since the largest disk has to be returned to the first peg, each disk has to be moved at least twice. We proceed by induction on $n\ge 2$. For $n=2$, it is easy to check that the two minimal sequences of moves are of length $7$, that are
$$
12 - 13 - 12 - 23 - 13 - 12 - 13 \quad\text{and}\quad 13 -12 - 13 - 23 - 12 - 13 - 12.
$$
Suppose that the result is true for transferring $n-1$ disks from one peg to the same peg. First, if we want to transfer the $n$th disk from Peg~1 to another peg, Peg~$i$ with $\{i,j\}=\{2,3\}$, the $n-1$ smallest disks have to be transferred from Peg~1 to Peg~$j$ using at least $2^{n-1}-1$ moves. Then, we distinguish two cases.
\begin{case}
If we want to transfer the largest disk from Peg~$i$ to Peg~1, the $n-1$ smallest disks have to be transferred from Peg~$j$ to Peg~$j$ using at least $2^n-1$ moves by induction hypothesis. Finally, we transfer the $n-1$ smallest disks from Peg~j to Peg~1 with at least $2^{n-1}-1$ moves. So, the number of moves is at least $2^{n+1}-1$ if we follow this strategy.
\end{case}
\begin{case}
If we want to transfer the largest disk from Peg~$i$ to Peg~$j$, the $n-1$ smallest disks have to be transferred from Peg~$j$ to Peg~1 using at least $2^{n-1}-1$ moves. After that, for transferring the largest disk from Peg~$j$ to Peg~1, the $n-1$ smallest disks have to be transferred from Peg~1 to Peg~$i$ using at least $2^{n-1}-1$ moves. Finally, we transfer the $n-1$ smallest disks from Peg~$i$ to Peg~1 using at least $2^{n-1}-1$ moves. Thus, the number of moves is at least $2^{n+1}-1$ if we follow this strategy.
\end{case}
In all cases, we have proved that the minimal number of moves for transferring $n$ disks from Peg~1 to Peg~1 is exactly $2^{n+1}-1$, following the two possible strategies that have been represented in Figure~\ref{fig:1}.

If we only require that the smallest disk, instead of the largest disk, has to be moved, the minimal number of moves for transferring $n\ge 2$ disks from one peg to the same peg is seven. The result is obtained when we only move the two smallest disks and let the $n-2$ largest disks unmoved on the starting peg.
\end{proof}

\begin{figure}
\begin{center}
\caption{Two possible strategies for transferring $n$ disks from one peg to the same peg.}\label{fig:1}
\vspace{1 mm}
\begin{tabular}{|c|c|}
\hline
Position & Number \\  & of Moves \\
\hline
\begin{tikzpicture}[scale=0.3]
\draw (0,0) -- (14,0);
\draw (2.5,0) -- (2.5,3);
\draw (7,0) -- (7,3);
\draw (11.5,0) -- (11.5,3);
\draw[fill=graygreen] (0.5,0) rectangle (4.5,0.5);
\draw[fill=graygreen] (1,0.5) rectangle (4,1);
\draw[fill=graygreen] (1.5,1) rectangle (3.5,1.5);
\draw[fill=graygreen] (2,1.5) rectangle (3,2);
\end{tikzpicture}
& 0 \\
\hline
\begin{tikzpicture}[scale=0.3]
\draw (0,0) -- (14,0);
\draw (2.5,0) -- (2.5,3);
\draw (7,0) -- (7,3);
\draw (11.5,0) -- (11.5,3);
\draw[fill=graygreen] (5,0) rectangle (9,0.5);
\draw[fill=graygreen] (10,0) rectangle (13,0.5);
\draw[fill=graygreen] (10.5,0.5) rectangle (12.5,1);
\draw[fill=graygreen] (11,1) rectangle (12,1.5);
\end{tikzpicture}
& $2^{n-1}$ \\
\hline
\begin{tikzpicture}[scale=0.3]
\draw (0,0) -- (14,0);
\draw (2.5,0) -- (2.5,3);
\draw (7,0) -- (7,3);
\draw (11.5,0) -- (11.5,3);
\draw[fill=graygreen] (0.5,0) rectangle (4.5,0.5);
\draw[fill=graygreen] (10,0) rectangle (13,0.5);
\draw[fill=graygreen] (10.5,0.5) rectangle (12.5,1);
\draw[fill=graygreen] (11,1) rectangle (12,1.5);
\end{tikzpicture}
& $2^{n}$ \\
\hline
\begin{tikzpicture}[scale=0.3]
\draw (0,0) -- (14,0);
\draw (2.5,0) -- (2.5,3);
\draw (7,0) -- (7,3);
\draw (11.5,0) -- (11.5,3);
\draw[fill=graygreen] (0.5,0) rectangle (4.5,0.5);
\draw[fill=graygreen] (1,0.5) rectangle (4,1);
\draw[fill=graygreen] (1.5,1) rectangle (3.5,1.5);
\draw[fill=graygreen] (2,1.5) rectangle (3,2);
\end{tikzpicture}
& $2^{n-1}-1$ \\
\hline
\hline
Total & $2^{n+1}-1$ \\
\hline
\end{tabular}\quad
\begin{tabular}{|c|c|}
\hline
Position & Number \\ & of Moves \\
\hline
\begin{tikzpicture}[scale=0.3]
\draw (0,0) -- (14,0);
\draw (2.5,0) -- (2.5,3);
\draw (7,0) -- (7,3);
\draw (11.5,0) -- (11.5,3);
\draw[fill=graygreen] (0.5,0) rectangle (4.5,0.5);
\draw[fill=graygreen] (1,0.5) rectangle (4,1);
\draw[fill=graygreen] (1.5,1) rectangle (3.5,1.5);
\draw[fill=graygreen] (2,1.5) rectangle (3,2);
\end{tikzpicture}
& 0 \\
\hline
\begin{tikzpicture}[scale=0.3]
\draw (0,0) -- (14,0);
\draw (2.5,0) -- (2.5,3);
\draw (7,0) -- (7,3);
\draw (11.5,0) -- (11.5,3);
\draw[fill=graygreen] (5,0) rectangle (9,0.5);
\draw[fill=graygreen] (10,0) rectangle (13,0.5);
\draw[fill=graygreen] (10.5,0.5) rectangle (12.5,1);
\draw[fill=graygreen] (11,1) rectangle (12,1.5);
\end{tikzpicture}
& $2^{n-1}$ \\
\hline
\begin{tikzpicture}[scale=0.3]
\draw (0,0) -- (14,0);
\draw (2.5,0) -- (2.5,3);
\draw (7,0) -- (7,3);
\draw (11.5,0) -- (11.5,3);
\draw[fill=graygreen] (9.5,0) rectangle (13.5,0.5);
\draw[fill=graygreen] (1,0) rectangle (4,0.5);
\draw[fill=graygreen] (1.5,0.5) rectangle (3.5,1);
\draw[fill=graygreen] (2,1) rectangle (3,1.5);
\end{tikzpicture}
& $2^{n-1}$ \\
\hline
\begin{tikzpicture}[scale=0.3]
\draw (0,0) -- (14,0);
\draw (2.5,0) -- (2.5,3);
\draw (7,0) -- (7,3);
\draw (11.5,0) -- (11.5,3);
\draw[fill=graygreen] (0.5,0) rectangle (4.5,0.5);
\draw[fill=graygreen] (5.5,0) rectangle (8.5,0.5);
\draw[fill=graygreen] (6,0.5) rectangle (8,1);
\draw[fill=graygreen] (6.5,1) rectangle (7.5,1.5);
\end{tikzpicture}
& $2^{n-1}$ \\
\hline
\begin{tikzpicture}[scale=0.3]
\draw (0,0) -- (14,0);
\draw (2.5,0) -- (2.5,3);
\draw (7,0) -- (7,3);
\draw (11.5,0) -- (11.5,3);
\draw[fill=graygreen] (0.5,0) rectangle (4.5,0.5);
\draw[fill=graygreen] (1,0.5) rectangle (4,1);
\draw[fill=graygreen] (1.5,1) rectangle (3.5,1.5);
\draw[fill=graygreen] (2,1.5) rectangle (3,2);
\end{tikzpicture}
& $2^{n-1}-1$ \\
\hline
\hline
Total & $2^{n+1}-1$ \\
\hline
\end{tabular}
\end{center}
\end{figure}

Now, we estimate the minimal number of moves for winning the normal play two-player Tower of Hanoi. In case the game is drawn, then we say that the minimal number of moves is infinite.

\begin{thm}
Let $\MM_l(n)$ denote the minimal number of moves needed for winning a normal play game on $l\ge 3$ pegs and $n\ge 1$ disks. Then,
$$
\begin{array}{ccl}
\MM_3(n) & = & \left\{\begin{array}{l}
\begin{array}{ll}
2^{n}-1, & \text{ for }n\ge 1,\text{ for (EC1,4)},\\
\end{array}\\
\left\{\begin{array}{ll}
2^{n+1}-1, & \text{for }n\ge 2,\text{ for (EC2)},\\
7, & \text{for }n\ge 2,\text{ for (EC3)},\\
\end{array}\right.\\
\left\{\begin{array}{ll}
2^n-1, & \text{for }n\le 2,\text{ for (EC5)},\\
7, & \text{for }n\ge 3,\text{ for (EC5)}.\\
\end{array}\right.\\ 
\end{array}\right.\\ \\
\MM_l(n) & = & \left\{\begin{array}{ll}
1, & \text{for } n=1, \ l\ge 4 \text{ and for (EC1,4,5)},\\
3, & \text{for } n=2, \ l\ge 4 \text{ and for (EC4,5)},\\
\infty, & \text{otherwise for }l\ge4.\\
\end{array}\right.
\end{array}
$$
\end{thm}

\begin{proof}
Apply the results in this section to the two-player setting.
\end{proof}

\subsection{The minimal number of moves for winning scoring play}

When the players move blindly (ignoring winning) and just follow the classical minimal algorithm, we obtain the total scores according to the following two lemmas.

\begin{lem}\label{lem:1}
Given $n \geq 1$ disks, three pegs and three weights of real numbers $w_{12}$, $w_{13}$, and $w_{23}$. Then, for the two-player weighted Tower of Hanoi game of transferring $n$ disks from Peg $1$ to Peg $3$ by the minimal algorithm, the total score is 
\begin{itemize}
\item
$\Delta_{13}(n) = w_{13}$ if $n$ is odd; 
\item
$\Delta_{13}(n) = w_{12}+w_{23} - w_{13}$ if $n$ is even.
\end{itemize}
\end{lem}

\begin{proof}
The statement is that \one\ gets $w_{13}$ (or $w_{12}+w_{23}-w_{13}$) more points than \two\ when $n$ is odd (or $n$ is even, resp.). We prove it by induction on $n$. 

When $n=1$, \one\ gets $\Delta_{13}(n)=w_{13}$ points by moving one disk from Peg~1 to Peg~3. When $n=2$, by using the usual minimal algorithm, the difference of points is $\Delta_{13}(n)=w_{12}+w_{23}-w_{13}$.

Now suppose that the statement is true for $n-1$. \one\ takes the strategy based on the minimal algorithm.

When $n$ is odd, then the movement of the smaller $n-1$ disks from Peg~1 to Peg~2 gives \one\ $\Delta_{12}(n-1) = w_{13}+w_{23}-w_{12}$ more points since $n-1$ is even and by the assumption of induction. Then, \two\ gets $w_{13}$ points by moving the largest disk from Peg~1 to Peg~3. Finally, the movement of the smaller $n-1$ disks from Peg~2 to Peg~3 gives \one\ $\Delta_{23}(n-1) = w_{12}+w_{13}-w_{23}$ more points. So, the difference of the total points is $\Delta_{13}(n)=(w_{13}+w_{23}-w_{12})-w_{13}+(w_{12}+w_{13}-w_{23})=w_{13}$, when $n$ is odd.

When $n$ is even, using the same argument, \one\ gets $w_{12}$ and $w_{23}$ more points in the phases of moving the smaller $n-1$ disks from Peg~1 to Peg~2 and Peg~2 to Peg~3, respectively by the assumption of induction. So, the difference of the total points is $\Delta_{13}(n) = w_{12}+w_{23}-w_{13}$.
\end{proof}

\begin{lem}
Given $n\geq 2$ disks, three pegs and three weights of real numbers $w_{12}$, $w_{13}$, and $w_{23}$. Then, for the two-player weighted Tower of Hanoi game of transferring $n$ disks from Peg $1$ to Peg $1$ by the minimal algorithm, if we suppose that the largest disk be moved, the total score is 
\begin{itemize}
\item
$\Delta_{11}(n) = 3w_{23}-w_{12}-w_{13}$ if $n$ is odd; 
\item
$\Delta_{11}(n) = w_{12}+w_{13} - w_{23}$ if $n$ is even.
\end{itemize}
\end{lem}

\begin{proof}
By induction on $n\ge2$. For $n=2$, we have already seen in the proof of Theorem~\ref{thm:1} that the two minimal sequences of moves for transferring two disks from Peg~1 to Peg~1 are
 $$
12 - 13 - 12 - 23 - 13 - 12 - 13 \quad\text{and}\quad 13 -12 - 13 - 23 - 12 - 13 - 12.
$$
In both cases, we obtain that $\Delta_{11}(2)=w_{12}+w_{13} - w_{23}$. Suppose now that the result is true for $n-1\ge2$. There are two minimal algorithms for transferring $n$ disks from Peg~1 to Peg~1. We recall them below. The first minimal algorithm is:
\begin{enumerate}[i)]
\item
the $n-1$ smallest disks are transferred from Peg~1 to Peg~3 using the $2^{n-1}-1$ moves of the minimal algorithm,
\item
the largest disk is transferred from Peg~1 to Peg~2,
\item
the $n-1$ smallest disks are transferred from Peg~3 to Peg~1 using the $2^{n-1}-1$ moves of the minimal algorithm,
\item
the largest disk is transferred from Peg~2 to Peg~3,
\item
the $n-1$ smallest disks are transferred from Peg~1 to Peg~2 using the $2^{n-1}-1$ moves of the minimal algorithm,
\item
the largest disk is transferred from Peg~3 to Peg~1,
\item
the $n-1$ smallest disks are transferred from Peg~2 to Peg~1 using the $2^{n-1}-1$ moves of the minimal algorithm.
\end{enumerate}
Thus, we have
$$
\Delta_{11}(n) \begin{array}[t]{l}
 = \Delta_{13}(n-1) -w_{12} + \Delta_{31}(n-1) - w_{23} + \Delta_{12}(n-1) - w_{13} + \Delta_{21}(n-1) \\
 = 2\left(\Delta_{13}(n-1) + \Delta_{12}(n-1)\right) - \left(w_{12}+w_{13}+w_{23}\right). \\
\end{array} 
$$
Then, from Lemma~\ref{lem:1}, if $n$ is even, we obtain that
$$
\Delta_{11}(n) = 2\left(w_{13} + w_{12}\right) - \left(w_{12}+w_{13}+w_{23}\right) = w_{12} + w_{13} - w_{23}
$$
and, if $n$ is odd, we have
$$
\Delta_{11}(n) = 2\left(\left(w_{12}+w_{23}-w_{13}\right) + \left(w_{13}+w_{23}-w_{12}\right)\right) - \left(w_{12}+w_{13}+w_{23}\right) = 3w_{23}-w_{12}-w_{13}.
$$
Finally, we consider the second minimal algorithm, that is:
\begin{enumerate}[i)]
\item
the $n-1$ smallest disks are transferred from Peg~1 to Peg~3 using the $2^{n-1}-1$ moves of the minimal algorithm,
\item
the largest disk is transferred from Peg~1 to Peg~2,
\item
the $n-1$ smallest disks are transferred from Peg~3 to Peg~3 using the $2^n-1$ moves of the minimal algorithm,
\item
the largest disk is transferred from Peg~2 to Peg~1,
\item
the $n-1$ smallest disks are transferred from Peg~3 to Peg~1 using the $2^{n-1}-1$ moves of the minimal algorithm.
\end{enumerate}
Thus, we have
$$
\Delta_{11}(n) = \Delta_{13}(n-1) - w_{12} + \Delta_{33}(n-1) - w_{12} + \Delta_{31}(n-1) = 2\Delta_{13}(n-1) + \Delta_{33}(n-1) - 2w_{12}.
$$
Then, from Lemma~\ref{lem:1} for $\Delta_{13}(n-1)$ and by induction hypothesis for $\Delta_{33}(n-1)$, if $n$ is even, we have
$$
\Delta_{11}(n) = 2w_{13} + \left(3w_{12}-w_{13}-w_{23}\right) -2w_{12} = w_{12} + w_{13} - w_{23}
$$
and, if $n$ is odd, we have
$$
\Delta_{11}(n) = 2\left(w_{12}+w_{23}-w_{13}\right) + \left(w_{13}+w_{23} - w_{12}\right) - 2w_{12} = 3w_{23}-w_{12}-w_{13}.
$$
This completes the proof.
\end{proof}

We estimate the minimal number of moves for winning scoring Tower of Hanoi under the five ending conditions (EC1-5). We recall that the starting peg is Peg~1 and for (EC1) we suppose the tower is transferred to Peg~3.

\begin{thm}\label{thm:7}
Let $\M(n)$ denote the minimal number of moves needed for a winning game on three pegs and $n\geq1$ disks. Then, for $n\leq2$,
\begin{itemize}
\item
for (EC1):
$$
\M(1) = \left\{\begin{array}{ll}
 1 & \text{ if } w_{13}\neq 0,\\
 \infty & \text{ otherwise},
\end{array}\right.
$$
$$
3 \leq \M(2) \leq \left\{\begin{array}{ll}
  3 & \text{ if } w_{12}+w_{23}>w_{13}, \\
  5 & \text{ if } 3w_{13}>w_{12}+w_{23}, \\
  \infty & \text{ otherwise}, \\
\end{array}\right.
$$ 
\item
for (EC2,3):
$$
\M(1) = \infty,
$$
$$
7 \leq \M(2) \leq \left\{\begin{array}{ll}
 7 & \text{ if } w_{12}+w_{13}>w_{23}, \\
 \infty & \text{ otherwise}, \\
\end{array}\right.
$$
\item
for (EC4,5):
$$
\M(1) = \left\{\begin{array}{ll}
 1 & \text{ if } \max\{w_{12},w_{13}\}\neq 0,\\
 \infty & \text{ otherwise},
\end{array}\right.
$$
$$
3 \leq \M(2) \leq \left\{\begin{array}{ll}
 3 & \text{ if } w_{12}+w_{23}>w_{13} \text{ or } w_{13}+w_{23}>w_{12}, \\
 5 & \text{ if } 3w_{13}>w_{12}+w_{23}  \text{ or } 3w_{12}>w_{13}+w_{23}, \\
 7 & \text{ if } w_{12}+w_{13}>w_{23}, \\
 \infty & \text{ otherwise.} \\
\end{array}\right.
$$
\end{itemize}
For $n\geq3$, if $w_{12}=w_{13}=w_{23}=\alpha$, we have
$$
\M(n) = \left\{\begin{array}{ll}
 2^{n}-1 & \text{ if } \alpha>0, \\
 \infty & \text{ if } \alpha\leq 0. \\
\end{array}\right.
$$
Otherwise, suppose that
$$
\beta_1 = \left\{\begin{array}{ll}
 w_{13} & \text{ if } n \text{ odd},\\
 w_{12}+w_{23}-w_{13} & \text{ if } n \text{ even}, \\
\end{array}\right.
$$
$$
\beta_2 = \left\{\begin{array}{ll}
 3w_{23}-w_{12}-w_{13} & \text{ if } n \text{ odd},\\
 w_{12}+w_{13}-w_{23} & \text{ if } n \text{ even}, \\
\end{array}\right.
$$
$$
\beta_3 = \left\{\begin{array}{ll}
\max\{w_{13}, w_{12}\} & \text{ if } n \text{ odd},\\
\max\{w_{12}+w_{23}-w_{13}, w_{13}+w_{23}-w_{12}\} & \text{ if } n \text{ even}, \\
\end{array}\right.
$$
$$
\text{and }\gamma = \max \left\{ w_{ij}+w_{ik}-2w_{jk} \ \middle\vert\ \{i,j,k\}=\{1,2,3\}  \right\}.
$$
Then, 
\begin{itemize}
\item
for (EC1):
$$
\M(n) = 2^{n}-1 \text{ if } \beta_1>0,
$$
otherwise,
$$
2^n \leq \M(n) \leq  2^{n}+15+16\left\lfloor\frac{-\beta_1}{2\gamma}\right\rfloor,
$$
except for $n=3$ with (i) $\gamma = w_{13}+w_{23}-2w_{12}$ and (ii) $\gamma = w_{12}+w_{13}-2w_{23}$, where
$$
8 \leq \M(n) \leq \min\biggl\{27+16\left\lfloor\frac{2(w_{12}+w_{23})-3w_{13}}{2\gamma}\right\rfloor, 29+16\left\lfloor\frac{-w_{13}}{2\gamma}\right\rfloor\biggr\},
$$
\item
for (EC2):
$$
\M(n) = 2^{n+1}-1 \text{ if } \beta_2 > 0, 
$$
$$
2^{n+1} \leq \M(n) \leq 2^{n+1}+15+16\left\lfloor\frac{-\beta_2}{2\gamma}\right\rfloor \text{ otherwise,}
$$
\item
for (EC3):
$$
\M(n) = \left\{\begin{array}{ll}
7 & \text{ if } w_{12}+w_{13}>w_{23},\\
15 & \text{ if } w_{12}+w_{13}\leq w_{23} \text{ and } \beta_2 > 0, \\
\end{array}\right.
$$
$$
16 \leq \M(n) \leq 
 31+16\left\lfloor\frac{-\beta_2}{2\gamma}\right\rfloor \text{ otherwise},
$$
\item
for (EC4):
$$
\M(n) = 2^n-1 \text{ if } \beta_3 > 0,
$$
otherwise,
$$
2^{n} \leq \M(n) \leq 
 \min\biggl\{
2^{n}+15+16\left\lfloor\frac{-\beta_3}{2\gamma}\right\rfloor,
2^{n+1}+15+16\left\lfloor\frac{-\beta_2}{2\gamma}\right\rfloor
\biggr\},
$$
except for (i) $n=3$ and $\gamma = w_{13}+w_{23}-2w_{12}$, where
$$
8 \leq \M(3) \leq 
 \min\biggl\{
23+16\left\lfloor\frac{-w_{12}}{2\gamma}\right\rfloor, 
27+16\left\lfloor\frac{2(w_{12}+w_{23})-3w_{13}}{2\gamma}\right\rfloor, 
29+16\left\lfloor\frac{-w_{13}}{2\gamma}\right\rfloor,$$
$$
31+16\left\lfloor\frac{-\beta_2}{2\gamma}\right\rfloor\biggr\}, 
$$
for (ii) $n=3$ and $\gamma = w_{12}+w_{23}-2w_{13}$, where
$$
8 \leq \M(3) \leq 
 \min\biggl\{
23+16\left\lfloor\frac{-w_{13}}{2\gamma}\right\rfloor, 
27+16\left\lfloor\frac{2(w_{13}+w_{23})-3w_{12}}{2\gamma}\right\rfloor, 
29+16\left\lfloor\frac{-w_{12}}{2\gamma}\right\rfloor,$$
$$
31+16\left\lfloor\frac{-\beta_2}{2\gamma}\right\rfloor\biggr\},
$$
and for (iii) $n=3$ and $\gamma = w_{12}+w_{13}-2w_{23}$, where
$$
8 \leq \M(3) \leq 
 \min\biggl\{
27+16\left\lfloor\frac{2(w_{12}+w_{23})-3w_{13}}{2\gamma}\right\rfloor, 
29+16\left\lfloor\frac{-w_{13}}{2\gamma}\right\rfloor,
$$
$$
27+16\left\lfloor\frac{2(w_{13}+w_{23})-3w_{12}}{2\gamma}\right\rfloor, 
29+16\left\lfloor\frac{-w_{12}}{2\gamma}\right\rfloor,
31+16\left\lfloor\frac{-\beta_2}{2\gamma}\right\rfloor\biggr\},
$$
\item
for (EC5): 
$$
7 \leq \M(n) \leq \left\{\begin{array}{ll}
7 & \text{ if } w_{12}+w_{13}>w_{23}
\text{ or } (n=3 \text{ and } \beta_3>0),\\
15 & \text{ if } (w_{12}+w_{13}\leq w_{23} \text{ and } \beta_2>0) \text{ or }  (n=4 \text{ and } \beta_3>0), \\
2^n-1 & \text{ if } \beta_3>0,\\
\end{array}\right.
$$
otherwise,
$$
8 \leq \M(n) \leq 
 \min\biggl\{
31+16\left\lfloor\frac{-\beta_2}{2\gamma}\right\rfloor,
2^{n}+15+16\left\lfloor\frac{-\beta_3}{2\gamma}\right\rfloor
\biggr\},
$$ 
except for the cases (i), (ii), and (iii) that are the same with (EC4).
\end{itemize}
\end{thm}

\begin{proof} 
We first prove the results for $n=1$ and $2$. For (EC1) and $n=1$, the number of moves for a winning game is 1 if $w_{13} \neq 0$ because what Anh can do in the first turn is just to  move the disk from Peg~1 to Peg~3 (note that if $w_{13} < 0$, Bao wins). For $n=2$ and if $w_{12}+w_{23}>w_{13}$, the minimal algorithm for transferring the two disks from Peg~1 to Peg~3 is used for Anh to win. For $n=2$ and if $3w_{13}>w_{12}+w_{23}$, the move sequence $13-12-13-23-13$ with $\Delta(2)=3w_{13}-w_{12}-w_{23}$ is used. For (EC2,3), in the case of $n=1$, the game cannot be started because the disk is not allowed to be moved from Peg~1. For $n=2$ and if $w_{12}+w_{13}>w_{23}$, one of the sequences with seven moves used in the proof of Theorem \ref{thm:5} is adopted, resulting in $\Delta(2)=w_{12}+w_{13}-w_{23}$ (indeed these are minimal with respect to moves and leads to a win for \one ). For (EC4,5), the result for $n=1$ is obtained by combining the conditions for moving the disk to Peg~2 and to Peg~3. For $n=2$, the result is obtained by considering all the cases of Anh completing the tower on Peg~1, 2, and 3 and by employing the move sequences in (EC1-3).

For $n\geq 3$, if $w_{12}=w_{13}=w_{23}=\alpha\leq 0$, Anh can not get a positive score so she uses the strategy of escaping the game to end and achives a draw. If $\alpha >0$, Anh wins using the minimal algorithm with difference of score $\Delta(n)=\alpha >0$. Otherwise, that is, when $n\geq 3$ and not all $w_{12}$, $w_{13}$, and $w_{23}$ are equal, Anh uses the common strategy for all ending conditions (EC1-5) of using the minimal algorithm for transferring the tower from Peg 1 to a final peg if that algorithm can be regarded as the concatenation of the move sequences $s_1$ and $s_2^{-1}$ used in the proof of Theorem~\ref{thm:4}. For all (EC1-5) and $n\geq4$, and for most of the cases with $n=3$, this strategy works. Otherwise, that is, in some exceptional cases with $n=3$, we construct new move sequences. We state the detail for each of the ending conditions below.

For (EC1) and $n\geq3$, except for the two cases ($n=3$ and $\gamma = w_{13}+w_{23}-2w_{12}$) and ($n=3$ and $\gamma = w_{12}+w_{13}-2w_{23}$), the minimal algorithm for transferring the tower from Peg~1 to Peg~3 can be used as concatenation of $s_1$ and $s_2^{-1}$ because the minimal algorithm reaches the intermediate position satisfying the condition for $P$ in Theorem~\ref{thm:4} with an even number of moves. Then it is further divided into two subcases depending on whether $\Delta(n)$ is positive. When $\Delta(n) > 0$, the minimal algorithm results in Anh's win with $\M(n)=2^n-1$. Otherwise, the move sequence $s_3$ in the proof of Theorem~\ref{thm:4} is additionally used $\lambda$ times in the notation of Theorem~\ref{thm:4} for Anh to win, where by Lemma~\ref{lem:1},
$$
\lambda = \left\lfloor\frac{-\Delta(n)}{2\left(w_{ik}+w_{jk}-2w_{ij} \right)}\right\rfloor+1 =\left\lfloor\frac{-\beta_1}{2\gamma}\right\rfloor+1.
$$
Therefore, the minimal number of moves is bounded as
$$
\M(n)\leq 2^n-1+16\biggl(\left\lfloor\frac{-\beta_1}{2\gamma}\right\rfloor+1\biggr) =2^n+15+16\left\lfloor\frac{-\beta_1}{2\gamma}\right\rfloor.
$$
For $n=3$ and $\gamma = w_{13}+w_{23}-2w_{12}$, that is, if $w_{12}$ is the smallest weight, the minimal algorithm can not reach the position satisfying the condition for $P$, having the two smallest disks on Peg~3 and the largest disk on either of the remaining pegs. So the following two move sequences are considered instead as candidates for $s_1s_2^{-1}$. The sequence with a smaller number of moves is then actually used (we write the sequence in the format ($s_1$)-($s_2^{-1}$).
\begin{eqnarray*}
& &(12-13-23-12)-(23-13-12-23-12-13-23) \ \text{  (11 moves)}\\
& &(13-12-13-23-13-12)-(23-13-12-23-12-13-23) \ \text{  (13 moves)} 
\end{eqnarray*}
The total scores for these sequences are $\Delta(n)=2(w_{12}+w_{23})-3w_{13}$ and  $\Delta(n)=w_{13}$, respectively. Therefore, $\M(3)$ is bounded as stated. Next, for $n=3$ and $\gamma = w_{12}+w_{13}-2w_{23}$, similarly to the exceptional case just mentioned, the minimal algorithm can not be used as $s_1 s_2^{-1}$. So, the following two sequences are considered as candidates for $s_1$ and $s_2^{-1}$ and the one with a smaller number of moves is actually used.
\begin{eqnarray*}
& &(12-13-23-12-23-13-12-23)-(12-13-23) \ \text{  (11 moves)}\\
& &(12-13-23-12-13-23-13-12-13-23)-(12-13-23) \ \text{  (13 moves)} 
\end{eqnarray*}
For the first and the second sequences, the total scores are $\Delta(n)=2(w_{12}+w_{23})-3w_{13}$ and  $\Delta(n)=w_{13}$, respectively. Therefore, $\M(3)$ is bounded as stated. 

For (EC2) and $n\geq3$, the minimal algorithm for transferring the Tower from Peg~1 to itself, where the largest disk has to be moved, can be used as concatenation of $s_1$ and $s_2^{-1}$. More precisely, if $\gamma = w_{12}+w_{23}-2w_{13}$, the first minimal algorithm in Lemma~2 in which the role of Peg~2 and Peg~3 is exchanged is used. Then after the procedure ii) of the algorithm with an even number of moves, it reaches the intermediate position $P$, so this algorithm can be $s_1 s_2^{-1}$. Next, if $\gamma = w_{13}+w_{23}-2w_{12}$, the first minimal algorithm in Lemma~2 can be $s_1 s_2^{-1}$. Finally, if $\gamma = w_{12}+w_{13}-2w_{23}$, after the procedure iv) of the minimal algorithm with an even number of moves, it reaches the intermediate position $P$. So the minimal algorithm can be $s_1 s_2^{-1}$. Similarly to (EC1), it is further divided into two subcases depending on whether $\Delta(n)$ is positive. When $\Delta(n) > 0$, Anh wins with the minimal algorithm with $\M(n)=2^{n+1}-1$. Otherwise, by Lemma~2 the minimal number of moves is bounded as 
$$
\M(n)\leq 2^{n+1}-1+16\biggl(\left\lfloor\frac{-\beta_2}{2\gamma}\right\rfloor+1\biggr) = 2^{n+1}+15+16\left\lfloor\frac{-\beta_2}{2\gamma}\right\rfloor.
$$

For (EC3), recall that the Tower is moved to itself, but only the smallest disk has to be moved at least once. First, we examine the case when only the two smallest disks are to be moved during the game. Then Anh can win with the 7 moves in Table~1 if $\Delta(n)=w_{12}+w_{13}-w_{23}>0$. Next, when more than two disks are moved during the game, then Anh uses the minimal algorithm of transferring the smallest three disks from Peg~1 to itself with 15 moves. Then, as shown in (EC2) with $n=3$, if $\beta_2>0$, then $M(n)=2^{3+1}-1=15$ for all $n\geq 3$. Otherwise,
$$
\M(n)\leq 2^{3+1}-1+16\biggl(\left\lfloor\frac{-\beta_2}{2\gamma}\right\rfloor+1\biggr) = 31+16\left\lfloor\frac{-\beta_2}{2\gamma}\right\rfloor.
$$

For (EC4), in which one can win by completing the Tower at Peg~1, 2, or 3, the algorithms and results for (EC1) with Peg~2 or Peg~3 as the final peg and (EC2) are employed. First, recall that 
$$
\beta_3 = \left\{\begin{array}{ll}
\max\{w_{13}, w_{12}\} & \text{ if } n \text{ odd},\\
\max\{w_{12}+w_{23}-w_{13}, w_{13}+w_{23}-w_{12}\} & \text{ if } n \text{ even}. \\
\end{array}\right.
$$
When $\beta_3>0$, Anh can win by using the minimal algorithm to reach either of Peg~2 or Peg~3 with $2^n - 1$ moves. Otherwise, if pairs of $n$ and $\gamma$ are not the exceptional ones in (EC1) and (EC2), the minimal algorithms for reaching Peg~1, 2, or 3 with the repetitive part $s_3^{\lambda}$ is used. Then, the minimal number of moves is bounded as
$$
\M(n) \leq  \min\biggl\{2^{n}+15+16\left\lfloor\frac{-\beta_3}{2\gamma}\right\rfloor, 2^{n+1}+15+16\left\lfloor\frac{-\beta_2}{2\gamma}\right\rfloor\biggr\}
$$
as stated. The remaining cases are for (EC1) with Peg~2 or Peg~3 as the final peg with $n=3$ and with either of the following: (i) $\gamma = w_{13}+w_{23}-w_{12}$, (ii) $\gamma = w_{12}+w_{23}-w_{13}$, and (iii) $\gamma = w_{12}+w_{13}-w_{23}$. When $\gamma = w_{13}+w_{23}-w_{12}$, we evaluate the numbers of moves for each case of reaching Peg~3, Peg~2, and Peg~1. When Peg~3 is the final peg, this is exactly one of the exceptional cases in (EC1), so the number of moves should be
$$
\min\biggl\{27+16\left\lfloor\frac{2(w_{12}+w_{23})-3w_{13}}{2\gamma}\right\rfloor, 29+16\left\lfloor\frac{-w_{13}}{2\gamma}\right\rfloor\biggr\}.
$$
When Peg~2 is the final peg, this case of $\gamma = w_{13}+w_{23}-w_{12}$ is not at all exceptional; thus, the minimal algorithm for transferring the three disks from Peg~1 to Peg~2 is used as the sequence $s_1s_2^{-1}$. So, the number of moves is 
$$
2^3 - 1 +16\biggl(\left\lfloor\frac{-w_{12}}{2\gamma}\right\rfloor + 1\biggr)
= 23 +16\left\lfloor\frac{-w_{12}}{2\gamma}\right\rfloor.
$$
When Peg~1 is the final peg, there is no exception, so the minimal algorithm is used for $s_1s_2^{-1}$ and the number of moves is 
$$
2^{3+1} - 1 +16\biggl(\left\lfloor\frac{-\beta_2}{2\gamma}\right\rfloor + 1\biggr) = 31 +16\left\lfloor\frac{-\beta_2}{2\gamma}\right\rfloor.
$$
In all, the minimum of these numbers of moves is taken as the upper bound of $M_3(3)$ as stated for (i) of (EC4). For (ii), the argument is exactly the same with (i) by exchanging the role of Peg~2 and Peg~3. For (iii), since $\gamma = w_{12}+w_{13}-w_{23}$ has to be treated as an exceptional case for both Peg~2 and Peg~3 are final pegs, so the bound for $M_3(3)$ is obtained as stated.

Finally for (EC5), the algorithms and results used for (EC1) and (EC3) are employed for obtaining the bounds for $M_3(n)$. Since the argument is almost the same with (EC4), we omit the detail.
\end{proof}

\section{Discussion}\label{sec:disc}

We have studied a number of ending conditions for two-player variations of the classical Tower of Hanoi. In almost all these variations \two\ is at disadvantage, by his confinement to forced play, which gives the two-player game still the flavor of a one player game. Indeed each move by \two\ is automatic without thinking and his input to the game has diminished to a purely mechanical matter. 

Still \one 's optimal move paths in the scoring variation are non-trivial; in particular when it comes to minimizing the number of moves. Open problem: improve the bounds for the estimates of the minimal number of moves in Theorem~\ref{thm:7}.

What if we allow consecutive moves of one and the same disk?  In the two-player setting such variations requires some attention (in a one-player variation with real weights this might also require analysis). A least requirement is that a player may not return directly to the previous position.

\subsection{Two-player Tower of Hanoi with a ko-rule.}
Sometimes, in combinatorial games, to avoid trivial draws, a \emph{ko-rule} is added to the game (e.g. the game of Go): a player is not allowed to directly return to the previous game position. If some other move has been played between, the ko-rule is not invoked. So $G\rightarrow H\rightarrow G$ is not allowed, but $G\rightarrow H\rightarrow K\rightarrow G$ is.

However, if applied to a variation of two-player Tower of Hanoi where consecutive moves of the same disk is allowed, this ko-rule still gives a draw game from any non-trivial position, as we will next demonstrate. 

Study a position of distance 2 to a terminal position. Suppose, without loss of generality, that the second smallest disk is on Peg 2 and the smallest is on Peg 1 and all remaining disks are on Peg 3 (which is terminal). For the game to terminate, with a winner, the other player had to move the second largest disk on top of the third largest disk. However, this move need never be carried out (by either player), since there will always be one peg available for the smallest disk. In our example, if the smallest disk was previously on Peg 2, then move it to Peg 3; if it were on Peg 3, then move it to Peg 2. The analogous draw strategy can be applied to new distance-2-to-terminal positions ad infinitum. (Note that this analysis is independent of normal play or scoring play convention.)

It turns out that, by combining an extended ko-rule with some other features, two-player variations can become interesting for both players, and sometimes they even put Bao at advantage.

\subsection{When is Bao at advantage?}\label{S:4}

The analysis so far has been disadvantageous for \two, apart from very special circumstances where \one\ was forced to revert to infinite play. Let us discuss minor alternations of the game rules, which gives \two\ a greater impact on the game. We settle with the non-scoring variation, and first without consecutive moves on the same disk, and say terminal tower on Peg 3. 

A standard variation for impartial games is to play mis\`{e}re rather than normal play, that is, the player who terminates the game loses. Since consecutive moves on the same disk are not allowed, Anh moves always the smallest disk, so she will lose if she puts it on top of the rest of the disks. However, she will never do this, since the smallest disk, in this variation always has two move options. Hence, we impose a rule which avoids a trivial draw in mis\'ere play. A standard technique in combinatorial games is to use `compulsory' moves. Here: a player  must put a disk on a one size larger disk if such a move is available. This rule makes mis\`{e}re meaningful, because a final move is forced if the game goes this far. However there is no guarantee that the game will reach the second last position. In fact, since \one\ still controls the game by always moving the smaller disk she will have to avoid that the game continues until the second last position. Now, the question arises whether she can, at least force drawn. Of course now Anh only has a choice every second move, and the analysis is straight forward. For 2 disks, she can draw, by playing smallest disk to Peg 3, and then in her second move she plays it on Peg 2, which is a forced move. 

Another variation which gives \two\ more impact on the game, is to allow more than one move on a given disk in a direct sequence of moves. If the previous player moved disk $d$, then the current player can move it as well, unless the position would become exactly the same as another position in the current `round'. Here `round' stands for a circuit where only the given disk has been moved. Once another disk has been moved the current circuit is broken, and a new `round' starts. For the three peg case, this gives an immediate advantage for \two, but he will only be able to use it to play a draw game in the general case.

It turns out, that combining the two modifications in this section, but playing normal (EC5), \two\ can force a win. Let us explain in the three-peg case and with two disks. \one\ plays the smallest disk to Peg~3, say. \two\ moves it to Peg~2, which is legal in this round (note that he is not forced to put it on top of the larger disk on Peg~1, since this move is not legal). Now, \one\ has to play the larger disk to Peg~3 and hence \two\ wins by putting the smaller disk on top of the larger. Note that the forced moves were not important in the case for two disks. Now, we let them play the three-disk case, still on three pegs. The move sequence will, for example, be $\emph{12}-23-\emph{12}-23-\emph{13}-12-\emph{23}-12$ where \one's moves are all forced (except the first one), and \two\ wins. \two\ cannot immediately adapt this strategy for the cases $n> 3$, because if $n=4$, then \one\ would get a parity advantage. We leave it as an open question to resolve this game in general.

\subsection{Disjunctive sum play}
A complete theory of disjunctive sum of cyclic impartial games, on finite numbers of positions, has recently been developed \cite{FY, Sm}, generalizing the classical theory of Sprague and Grundy (a disjunctive sum of games consists of a finite number of  component games, and at each stage play is in one of the components). In a disjunctive sum of Tower of Hanoi games, \one\ would no longer control all moves, because the move sequence is not necessarily alternating in each game component. Therefore, in spite of the simple solutions of the single component normal play games, disjunctive sum theory should offer new insights. We also note that it is non-trivial to count the minimal number of moves between two arbitrary Tower of Hanoi positions on three pegs \cite{Hi, W}, so we guess that a computation of generalized Sprague-Grundy values of the game will be difficult. 

Disjunctive sum theory for scoring games has been developed in various settings e.g. \cite{E1, LNS, LNNS, M, St}, but none yet for cyclic scoring games; in particular it will generalize the class well tempered games \cite{WJ} (since there is always an odd number of moves). 

\subsection{How about several players?}

Yet another variation is to play our standard variation with several players. In the case for (EC1), three players, three pegs and two disks, then the first player cannot win under optimal play, but she can decide which one of the other players that will win. If this game is played with three disks, then it is drawn. This can be seen this way. No player will be forced to put the second smallest disk on a final peg and thereby giving the player just after the winning move. This follows because the player who would have moved the supposedly second last move of the smallest disk, can choose to put it on either the second or third smaller disk, either of which prevents a bad forced move of the second smallest disk. Open problems: classify the variations of normal and scoring play where several player Tower of Hanoi terminates.\\

\noindent{\bf Acknowledgement.} We thank the anonymous referees for their comments and suggestions.

\end{document}